\documentclass[conference]{IEEEtran}
\usepackage{amssymb}
\usepackage[cmex10]{amsmath}
\usepackage{stfloats}
\usepackage{graphicx}
\usepackage{subfigure}
\usepackage{tabularx}
\usepackage{epsfig,epsf,color,balance,cite}
\usepackage{verbatim}
\usepackage{bm}
\usepackage{amsmath}
\usepackage{hyperref}
\usepackage{times}
\usepackage{amsthm}
\newtheorem{theorem}{Theorem}

\usepackage{algorithm}
\usepackage{algorithmicx}
\usepackage{algpseudocode}

\ifCLASSINFOpdf
\else
\fi
\begin{document}
%
\title{Performance Analysis of User-centric Virtual Cell Dense Networks over mmWave Channels}

\author{\IEEEauthorblockN{Jianfeng Shi, Yinlu Wang, Hao Xu and Ming Chen}
\IEEEauthorblockA{National Mobile Communications Research Laboratory\\
Southeast University, Nanjing, China\\
Email: \{shijianfeng, yinluwang, \\xuhao2013, chenming\}@seu.edu.cn}
\and
\IEEEauthorblockN{Benoit Champagne}
\IEEEauthorblockA{Department of Electrical and Computer Engineering\\
	 McGill University, Montreal, Canada\\
Email: benoit.champagne@mcgill.ca}
}


%


\maketitle

\begin{abstract}
This paper analyzes the ergodic capacity of a user-centric virtual cell (VC) dense network, where multiple access points (APs) form a VC for each user equipment (UE) and transmit data cooperatively over millimeter wave (mmWave) channels.
Different from traditional microwave radio communications, blockage phenomena have an important effect on mmWave transmissions.
Accordingly, we adopt a distance-dependent line-of-sight (LOS) probability function and model the locations of the LOS and non-line-of-sight (NLOS) APs as two independent non-homogeneous Poisson point processes (PPP).
Invoking this model in a VC dense network, new expressions are derived for the downlink ergodic capacity, accounting for: blockage, small-scale fading and AP cooperation.
In particular, we compare the ergodic capacity for different types of fading distributions, including Rayleigh and Nakagami.
Numerical results validate our analytical expressions and show that AP cooperation can provide notable capacity gain, especially in low-AP-density regions.
\end{abstract}


%
\IEEEpeerreviewmaketitle

\section{Introduction}
\label{sec_intdn}
The demands of higher data rates for local areas services increase significantly, which trigger interests in research for the more spectral efficiency (SE) and energy efficiency (EE) system \cite{shi2017power,pan2017joint,Pan2017user}.
To overcome the situation, user-centric virtual cell (VC) networking has been advocated as one of the key breakthrough technologies for the fifth generation (5G) wireless networks \cite{boccardi2014five}.
In a user-centric VC network, several access points (APs) are distributed over a given coverage area and connected to a central controller via high-speed links.
Each user equipment (UE) is served by its surrounding APs in a cooperative way.
Besides, since the current microwave radio spectrum (from 300 MHz to 6 GHz) is scarce, it becomes vital to exploit the millimeter wave (mmWave) spectrum band (from 28 to 300 GHz).
In user-centric networks, the massive deployment of small VCs renders the short-range mmWave technologies very attractive \cite{bogale2016massive}.

The performance of user-centric networks using microwave transmissions was extensively studied in \cite{Jia2016performance,yang2016performance,Zhao2015cluster,wang2016downlink,Khan2015performance,Peng2014ergodic}, where
the expressions for the outage probability, coverage probability, successful access probability (SAP), SE, EE, and ergodic rate were deduced.
Specifically, by modeling the locations of APs as a marked poisson point process, closed-form expressions for the coverage probability, SE and EE were derived in \cite{Jia2016performance}.
To obtain tractable analytical expressions for the outage probability and ergodic rate, the Gauss-Chebyshev integration technique was applied in  \cite{yang2016performance}.
An explicit expression for the SAP was derived in \cite{Zhao2015cluster}, where SAP is defined as the conditional probability that the signal-to-interference-plus-noise ratio (SINR) exceed a threshold, given a predetermined set of serving APs.
In \cite{wang2016downlink}, the authors investigated the achievable ergodic
rate of each user in a VC-based distributed antenna system, where the users and antennas are uniformly distributed.
The downlink performance of a cloud radio access network (CRAN) with randomly distributed multiple antenna APs was investigated in \cite{Khan2015performance}, where closed-form expressions (either exact or approximate) for the outage probability were derived for three different transmission and AP selection schemes, differing in the number and choice of APs used to serve a particular user.
Under similar system model assumptions as in \cite{Khan2015performance}, a closed-form expression for the uplink ergodic capacity\footnote{Generally, ergodic capacity is studied based on the
	assumption that channel fading transitions through all possible
	fading states, and therefore this definition may not be practical for source transmission with fixed delay constraints\cite{Choudhury2007information}.} was derived in \cite{Peng2014ergodic}.

In contrast to traditional microwave transmissions, mmWave communications
exhibit strong directionality and suffer from severe path loss \cite{andrews2017modeling}.
Thus, the analytical expressions and methods in \cite{Jia2016performance,yang2016performance,Zhao2015cluster,Khan2015performance,Peng2014ergodic} cannot be applied directly in mmWave networks.
The performance of mmWave cellular networks was studied in \cite{singh2015tractable,bai2015coverage,he2016performance,Maamari2016coverage}.
Specifically, \cite{singh2015tractable} and \cite{bai2015coverage} proposed general tractable models to characterize the coverage and rate distribution in mmWave cellular networks with and without self-backhauling, respectively.
Expressions for the SINR and rate coverage probability were also derived as a function of the antenna geometry and base station density.
\cite{he2016performance} analyzed the outage performance of a mmWave CRAN and compared the ergodic capacity of line-of-sight (LOS) and non-line-of-sight (NLOS) APs.
The coverage probability in the downlink of mmWave heterogeneous networks with AP cooperation was recently studied in \cite{Maamari2016coverage}.
While considering a general multiple-input multiple-output (MIMO) transceiver model, integral expressions for the coverage probabilities are developed for the single antenna case.
Furthermore, while different fading distributions are considered for the desired links, Rayleigh fading is assumed for the interference links.


In this paper, motivated by these considerations, we analyze the ergodic capacity performance of a user-centric VC dense network operating in the mmWave spectrum band.
Specifically, we focus on the downlink VC network, where a typical UE chooses a fixed number of closest APs to form its VC.
Considering several important aspects, such as the AP location randomness, distance-dependent path loss, small-scale fading, directional beamforming and the AP cooperation, we derive the analytical expressions for the ergodic capacity in such networks.
More specifically, the ergodic capacity is analyzed for three types of small-scale fading distributions (i.e., Nakagami, Rayleigh and no fading) via stochastic geometry.
Numerical results validate our analytical expressions and show that the AP cooperation can provide significant capacity gain in a low-AP-density region.

In Section \ref{sec_sys}, we introduce the downlink user-centric VC network under study and specify the path loss, beamforming and SINR models within the mmWave framework.
Section \ref{sec_perays} analyzes the ergodic capacity of a typical UE under the three considered fading distributions.
Supporting numerical results along with discussions are presented in Section \ref{sec_numrst}.
Finally, Section \ref{sec_clsn} concludes the paper.

\section{System Model}
\label{sec_sys}
\subsection{User-centric VC Network}
We consider the downlink of a user-centric VC dense network, where a typical UE located at the origin\footnote{When the UEs are distributed as an independent stationary point process, the ergodic capacity of the typical UE located at the origin is identical to that of other UEs in the network \cite{zhang2015downlink,bai2015coverage}.}
is surrounded by multiple APs deployed according to a two dimensional homogeneous Poisson point process (PPP), i.e., $\Phi=\{X_k,~\forall k\}$ with density $\lambda$.
Let $r_k=|X_k|$ denote the distance between the UE and the $k$th AP. Without loss of generality, we let the APs be indexed in increasing order of distance, i.e. $r_k < r_l$ for $k<l$.
All the APs are connected to a central controller via high speed dedicated links (e.g., fiber optics)
and share the same resources (time or frequency) to transmit data.
Based on external measurements, the UE is assumed to select the $K$ closest APs to form its VC, as illustrated in Fig.~\ref{fig_sysmod}.
Hence, the typical UE is served by its corresponding set of APs, $\mathcal{V}_0 = \{1,2,\cdots,K\}$,
while it suffers interference from the APs which do not belong to $\mathcal{V}_0$.
Because of the blockage effect, an AP can be either LOS or NLOS to the UE.
However, with dense deployment, it is reasonable to assume that the link between any serving AP to the UE is LOS\footnote{This assumption will be verified by the simulation results in Section \ref{sec_numrst}.}.

\begin{figure}
	\centering
	\includegraphics[width=3.2in]{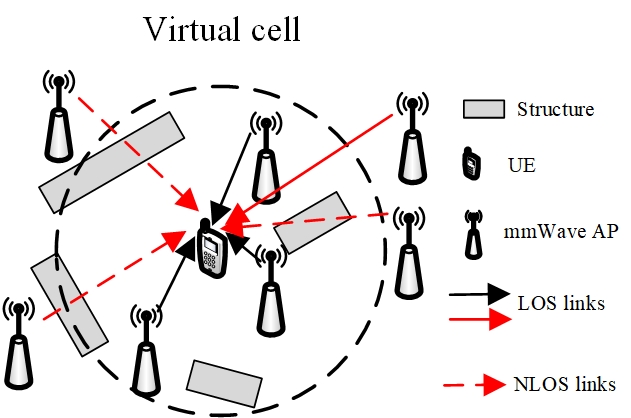}
	\caption{Illustration of a downlink VC network with $K=3$ AP selected. Black and red arrows refer to desired and interference links, respectively.\label{fig_sysmod}}
\end{figure}

\subsection{Transmission Model}
\label{subsec_path_beam}
\subsubsection{Small-scale Fading}
Let $\xi_k$ denote the small-scale complex fading coefficient betweeh the UE and the $k$th AP.
In this paper, we consider three small-scale fading distributions.
Firstly, in the case that the fading is Nakagami distributed, $|\xi_k|^2$ is a normalized Gamma random variable (r.v.) with $|\xi_k|^2\sim \Gamma(N_{\rm L},1/N_{\rm L})$ when the $k$th link is LOS and $|\xi_k|^2\sim \Gamma(N_{\rm N},1/N_{\rm N})$ when it is NLOS \cite{bai2015coverage}.
Here, $N_{\rm L}$ and $N_{\rm N}$ are the Nakagami parameters for LOS and NLOS links, respectively.
Secondly, under the Rayleigh fading distribution, $\xi_k$ is modeled as a zero mean complex Gaussian r.v. and $|\xi_k|^2\sim\exp(\mu)$, where $\mu$ is the parameter of the exponential distribution.
Thirdly, without small-scale fading, we ignore $\xi_k$ by letting $|\xi_k|^2=1$ in the corresponding formulas.

\subsubsection{Path Loss}
In mmWave communications, the path loss effects are quite different from those occurring in the traditional microwave radio band, due to the more serious absorption and blockage effects.
Define the LOS probability function $p(r)$ as the probability that a link of length $r$ is LOS.
Then, the NLOS probability of a link is $1-p(r)$.
According to \cite{bai2014analysis}, we assume that the blockages are modeled as a rectangle Boolean scheme.
That is, $p(r)=e^{-\beta r}$, where $\beta$ is the blockage parameter determined by the density and average size of the blockages.
Given a link of length $r$, its path loss function $L(r)$ is given by
\begin{equation}\label{equ_pathloss}
L(r)=\left\{
\begin{aligned}
&C_{\rm L}r^{-\alpha_{\rm L}},~\text{with probability}~ p(r) \\
&C_{\rm N}r^{-\alpha_{\rm N}},~\text{with probability} ~1-p(r),
\end{aligned}
\right.
\end{equation}
where $\alpha_{\rm L}$ and $\alpha_{\rm N}$ are the LOS and NLOS path loss exponents.
$C_{\rm L}$ and $C_{\rm N}$ are the intercepts of the LOS and NLOS path loss function, respectively.

\subsubsection{Directional Beamforming}
We assume that all APs are equipped with directional antennas, which approximatively follow a sectored antenna model.
For simplicity, we assume that the APs all have the same beamwidth, which is denoted by $\theta_b$.
Then, the antenna gain for a mmWave AP can be written as follows\cite{singh2015tractable},
\begin{equation*}
A_b(\theta)=\begin{cases}
M,\quad \text{if}\; |\theta|<\theta_b/2\\
m,\quad \text{otherwise},
\end{cases}
\end{equation*}
where $\theta$ is the angle of departure measured from boresight direction, $M$ is the main lobe gain and $m$ is the sidelobe gain.
The antenna gain $A_u(\theta)$ at UE side can be modeled in the same manner, whereas we assume omni-directional antennas for the users and $A_u(\theta)=1$ in this paper.
Assume that the serving APs for the typical UE adjust their beam angles to achieve the maximum antenna gains.
Then, we have $G_{k}=M, \forall k\in \mathcal{V}_0$, where $G_{k}$ is the antenna gain of the $k$th AP.
Furthermore, for the interference links, the interfering APs' angles follow the independent uniform distribution in $(-\pi,\pi]$.
As a result, the antenna gain of the interference link $G_{l}, \forall l \notin \mathcal{V}_0$ is a discrete r.v., whose probability distribution is given by
$\mathbb{P}(G_l=a_n)=b_n,\; n=1,2$.
Here, $\mathbb{P}(\cdot)$ denotes the probability of an event, $a_1 = M$, $b_1=\theta_b/{2\pi}$, $a_2=m$ and $b_2=1-\theta_b/{2\pi}$.

\subsubsection{SINR}
%
We assume that all APs transmit with the same power $P_t$.
Then, the SINR at the typical UE can be expressed as \cite{lin2014downlink}
\begin{align}\label{SINR}
\gamma&=\frac{P}{P^{I}+\sigma^2} \nonumber\\
&=\frac{\sum_{k\in \mathcal{V}_0}G_kL(r_{k})|\xi_{k}|^2}{\sum_{l\notin \mathcal{V}_0}G_{l}L(r_{l})|\xi_{l}|^2+\sigma^2},
\end{align}
where $P$, $P^{I}$ and $\sigma^2$ are the desired signal, interference and background noise power at the UE, respectively.
Note that these quantities are normalized by the transmitting power $P_t$ and the SINR in (\ref{SINR}) is a r.v., because of the randomness in the antenna gain $G_{l}$, distance $r_{l}$ and small-scale fading $\xi_{l}$.

\section{Performance Analysis}
\label{sec_perays}
In this section, after introducing preliminary mathematical notions and definitions, 
we derive the ergodic capacity for the typical UE in a downlink VC network under under three small-scale fading distributions, i.e., Nakagami, Rayleigh and no fading.

\subsection{Preliminaries}

Let $\Phi_{\rm L}$ and $\Phi_{\rm N}$ be the point process of the LOS and NLOS APs, respectively.
With negligible loss in accuracy \cite{bai2015coverage}, $\Phi_{\rm L}$ and $\Phi_{\rm N}$ can be modeled as two independent non-homogeneous PPP with density function $\lambda p(r)$ and $\lambda (1-p(r))$, respectively.
Then, the SINR can be reformulated into
\begin{align}\label{SINR_ref}
\gamma=\frac{\sum_{k\in\Phi\cap\mathcal{B}(r_K)}MC_{\rm L}r_k^{-\alpha_{\rm L}}|\xi_{k}|^2}{I_{\rm L}+I_{\rm N}+\sigma^2},
\end{align}
where $\mathcal{B}(r_K)$ denotes the circle centered at the origin of radius $r_K$, $I_{\rm L}=\sum_{l\in \Phi_{\rm L}\cap\bar{\mathcal{B}}(r_K)}G_{l}C_{\rm L}r_l^{-\alpha_{\rm L}}|\xi_{l}|^2$ and $I_{\rm N}=\sum_{l\in \Phi_{\rm N}\cap\bar{\mathcal{B}}(r_K)}G_{l}C_{\rm N}r_l^{-\alpha_{\rm N}}|\xi_{l}|^2$ are the interference powers from the LOS and NLOS APs, and
$\bar{\mathcal{B}}(r_K)$ represents the region outside of $\mathcal{B}(r_K)$.

Accordingly, the ergodic capacity (bps/Hz) of the typical UE is defined as
\begin{equation}\label{def_C}
C\triangleq\mathbb{E}_{\mathbf{G},\mathbf{r},\bm \xi}[\log_2(1+\gamma)],
\end{equation}
where $\mathbf G=\{G_l, \forall l\}$, $\mathbf r=\{r_l, \forall l\}$ and $\bm \xi=\{\xi_l, \forall l\}$.
Given that the serving APs of the typical UE are at distances of $r_1,\cdots,r_K$, the ergodic capacity can be rewritten as
\begin{equation}\label{ref_C}
C=\idotsint_{\mathcal{D}}C_{\text{cond}}(\bm{r})f_{\bm{r}}(\bm{r})\text{d}\bm{r},
\end{equation}
where the multiple integral domain is $\mathcal{D}=\{0<r_1\leqslant\cdots\leqslant r_K\}$, $\text{d}\bm{r}=\text{d}r_1\cdots \text{d}r_K$, $C_{\text{cond}}(\bm{r})$ is the conditional ergodic capacity.
$f_{(\bm{r})}(\bm{r})$ is the joint probability density function (PDF) of $r_1,\cdots,r_K$, given by \cite[(30)]{Moltchanov2012distance}
\begin{equation}\label{Jpdf_r1K}
f_{(\bm{r})}(\bm{r})=(2\pi\lambda)^Kr_1\cdots r_Ke^{-\pi\lambda r_K^2}.
\end{equation}


\subsection{Ergodic Capacity under Nakagami Fading}
In this subsection, we analyze the ergodic capacity when the small-scale fading is Nakagami distributed.
Recall that
$|\xi_l|^2$ is a normalized Gamma r.v. with $|\xi_l|^2\sim \Gamma(N_{\rm L},1/N_{\rm L})$ when the $l$th link is LOS and $|\xi_l|^2\sim \Gamma(N_{\rm N},1/N_{\rm N})$ when it is NLOS.
The ergodic capacity of user-centric VC dense network with Nakagami fading is given in Theorem \ref{teom_C_cond_Nakagami}.
\begin{theorem}\label{teom_C_cond_Nakagami}
	The conditional ergodic capacity with Nakagami small-scale fading is
	\begin{align}\label{fmula_C_cond_Naka}
	C_{\emph{cond}}(\bm{r})=&\frac{1}{\ln2}\int_0^\infty\frac{e^{-s\sigma^2}}{s}e^{-(Q_{\rm L}(s)+Q_{\rm N}(s))}\times\nonumber\\
	&\left(1-\prod_{k=1}^{K}(1-F(N_{\rm L},sMC_{\rm L}r_k^{-\alpha_{\rm L}}))\right)\emph{d} s,
	\end{align}
	where
	\begin{equation*}
	Q_{\rm L}(s)=2\pi\lambda \sum_{n=1}^{2}b_n\int_{r_K}^{\infty}F(N_{\rm L},sa_nC_{\rm L}x^{-\alpha_{\rm L}})p(x)x\emph{d}x,
	\end{equation*}
	\begin{equation*}
	Q_{\rm N}(s)=2\pi\lambda\sum_{n=1}^{2} b_n\int_{r_K}^{\infty}F(N_{\rm N},sa_nC_{\rm N}x^{-\alpha_{\rm N}})(1-p(x))x\emph{d}x,
	\end{equation*}
	$F(N,x)=1-1/(1+x/N)^N$, $a_n$ and $b_n$ are defined in Subsection \ref{subsec_path_beam}.
	By substituting (\ref{Jpdf_r1K}) and (\ref{fmula_C_cond_Naka}) into (\ref{ref_C}), the downlink ergodic capacity with Nakagami fading is given by
	\begin{align}\label{fmula_C_Naka}
	&C=\frac{(2\pi\lambda)^K}{\ln2}\idotsint_{\mathcal{D}}r_1\cdots r_Ke^{-\pi\lambda r_K^2}\int_0^\infty\frac{e^{-s\sigma^2}}{s}\times\nonumber\\
	&e^{-(Q_{\rm L}(s)+Q_{\rm N}(s))}
	\left(1\!\!-\!\!\prod_{k=1}^{K}(1\!\!-\!\!F(N_{\rm L},sMC_{\rm L}r_k^{-\alpha_{\rm L}}))\right)\emph{d} s \emph{d}\bm{r}.
	\end{align}
\end{theorem}
\begin{proof}
	The proof is given in Appendix \ref{proof_th1}.
\end{proof}

%

\subsection{Ergodic Capacity under Rayleigh Fading}
In this subsection, we analyze the ergodic capacity when the small-scale fading is Rayleigh distributed and $|\xi_l|^2\sim\exp(\mu)$.
The ergodic capacity of user-centric VC dense network with Rayleigh fading is given in Theorem \ref{teom_C_cond_Rayleigh}.
\begin{theorem}\label{teom_C_cond_Rayleigh}
	The conditional ergodic capacity with Rayleigh small-scale fading is
	\begin{align}\label{fmula_C_cond_Ray}
	C_{\emph{cond}}(\bm{r})=&\frac{1}{\ln2}\int_0^\infty\frac{e^{-s\sigma^2}}{s}e^{-(V_{\rm L}(s)+V_{\rm N}(s))}\times\nonumber\\
	&\left(1-\prod_{k=1}^{K}(1-H(\mu sMC_{\rm L}r_k^{-\alpha_{\rm L}}))\right)\emph{d} s,
	\end{align}
	where
	\begin{equation*}
	V_{\rm L}(s)=2\pi\lambda\sum_{n=1}^{2} b_n\int_{r_K}^{\infty}H(\mu sa_nC_{\rm L}x^{-\alpha_{\rm L}})p(x)x\emph{d}x,
	\end{equation*}
	\begin{equation*}
	V_{\rm N}(s)=2\pi\lambda\sum_{n=1}^{2} b_n\int_{r_K}^{\infty}H(\mu sa_nC_{\rm N}x^{-\alpha_{\rm N}})(1-p(x))x\emph{d}x,
	\end{equation*}
	$H(x)=1-1/(1+x)$, $a_n$ and $b_n$ are defined in Subsection \ref{subsec_path_beam}.
	By substituting (\ref{Jpdf_r1K}) and (\ref{fmula_C_cond_Ray}) into (\ref{ref_C}), the downlink ergodic capacity with Rayleigh fading is given by
	\begin{align}\label{fmula_C_Ray}
	&C=\frac{(2\pi\lambda)^K}{\ln2}\idotsint_{\mathcal{D}}r_1\cdots r_Ke^{-\pi\lambda r_K^2}\int_0^\infty\frac{e^{-s\sigma^2}}{s}\times\nonumber\\
	&e^{-(V_{\rm L}(s)+V_{\rm N}(s))}
	\left(1\!\!-\!\!\prod_{k=1}^{K}(1\!\!-\!\!H(\mu sMC_{\rm L}r_k^{-\alpha_{\rm L}}))\right)\emph{d} s\emph{d}\bm{r}.
	\end{align}
\end{theorem}
\begin{proof}
	The proof is given in Appendix \ref{proof_th2}.
\end{proof}

\subsection{Ergodic Capacity with no fading}
In this subsection, we study the ergodic capacity neglecting the small-scale fading.
The SINR can be rewritten as
\begin{align}\label{SINR_nofading}
\gamma'=\frac{\sum_{k\in\Phi\cap\mathcal{B}(r_K)}MC_{\rm L}r_k^{-\alpha_{\rm L}}}{I'_{\rm L}+I'_{\rm N}+\sigma^2},
\end{align}
where $I'_{\rm L}=\sum_{l\in \Phi_{\rm L}\cap\bar{\mathcal{B}}(r_K)}G_{l}C_{\rm L}r_l^{-\alpha_{\rm L}}$ and $I'_{\rm N}=\sum_{l\in \Phi_{\rm N}\cap\bar{\mathcal{B}}(r_K)}G_{l}C_{\rm N}r_l^{-\alpha_{\rm N}}$.
Then, the ergodic capacity of user-centric VC dense network is given in Theorem \ref{teom_C_cond_nofading}.
\begin{theorem}\label{teom_C_cond_nofading}
	The conditional ergodic capacity without considering small-scale fading is
	\begin{align}\label{fmula_C_cond_nofading}
	C_{\emph{cond}}(\bm{r})=&\frac{1}{\ln2}\int_0^\infty\frac{e^{-s\sigma^2}}{s}e^{-(W_{\rm L}(s)+W_{\rm N}(s))}\times\nonumber\\
	&\left(1-\prod_{k=1}^{K}e^{-sMC_{\rm L}r_k^{-\alpha_{\rm L}}}\right)\emph{d} s,
	\end{align}
	where
	\begin{equation*}
	W_{\rm L}(s)=2\pi\lambda\sum_{n=1}^{2} b_n\int_{r_K}^{\infty}\left(1-e^{-sa_nC_{\rm L}x^{-\alpha_{\rm L}}}\right)p(x)x\emph{d}x,
	\end{equation*}
	\begin{equation*}
	W_{\rm N}(s)=2\pi\lambda\sum_{n=1}^{2} b_n\int_{r_K}^{\infty}\left(1-e^{-sa_nC_{\rm N}x^{-\alpha_{\rm N}}}\right)(1-p(x))x\emph{d}x,
	\end{equation*}
	$a_n$ and $b_n$ are defined in Subsection \ref{subsec_path_beam}.
	By substituting (\ref{Jpdf_r1K}) and (\ref{fmula_C_cond_nofading}) into (\ref{ref_C}), the downlink ergodic capacity with no fading is given by
	\begin{align}\label{fmula_C_nofading}
	C=&\frac{1}{\ln2}\idotsint_{\mathcal{D}}(2\pi\lambda)^Kr_1\cdots r_Ke^{-\pi\lambda r_K^2}\int_0^\infty\frac{e^{-s\sigma^2}}{s}\times\nonumber\\
	&e^{-(W_{\rm L}(x)+W_{\rm N}(x))}
	\left(1-\prod_{k=1}^{K}e^{-sMC_{\rm L}r_k^{-\alpha_{\rm L}}}\right)\emph{d} s\emph{d}\bm{r}.
	\end{align}
\end{theorem}
\begin{proof}
	The proof is given in Appendix \ref{proof_th3}.
\end{proof}

\section{Numerical Results}
\label{sec_numrst}
In this section, we present selected numerical and simulation to validate our analysis in Section \ref{sec_perays}.
We consider the APs are distributed in a circular region with radius $R=100~\rm m$ and the UE is located at the origin.
The mmWave is assumed to be operated at 73 GHz and the bandwidth is $B=2\;\text{GHz}$.
The transmitting power of mmWave AP is $P_t=30$ dBm.
The normalized noise power is thus $\sigma^2\;(\text{dB})=-174+10\log_{10}(B)+10-P_t$.
The parameters of the directional AP antennas are set as follows:
main lobe gain $M=18$ dB, sidelobe gain $m=-2$ dB, and beamwidth $\theta_b=10^o$ \cite{singh2015tractable}.
Based on \cite{bai2015coverage}, the blockage parameter of LOS probability function $p(r)$ is set to a nominal value of $\beta=0.0071$.
The LOS and NLOS path loss exponents are chosen as $\alpha_{\rm L}=2$ and $\alpha_{\rm N}=4$, while the corresponding coefficients are $C_{\rm L}=C_{\rm N}=10^{-7}$.
The parameters of the Nakagami fading are $N_{\rm L}=3$ and $N_{\rm N}=2$.
The parameter of the Rayleigh fading is $\mu=1$ \cite{Maamari2016coverage}.
The analytical results are computed by numerical evaluation of the expressions derived in Section \ref{sec_perays}, while the simulation results are obtained by averaging over 1000 channel realizations.
\begin{figure}
	\centering
	\includegraphics[width=3.5in]{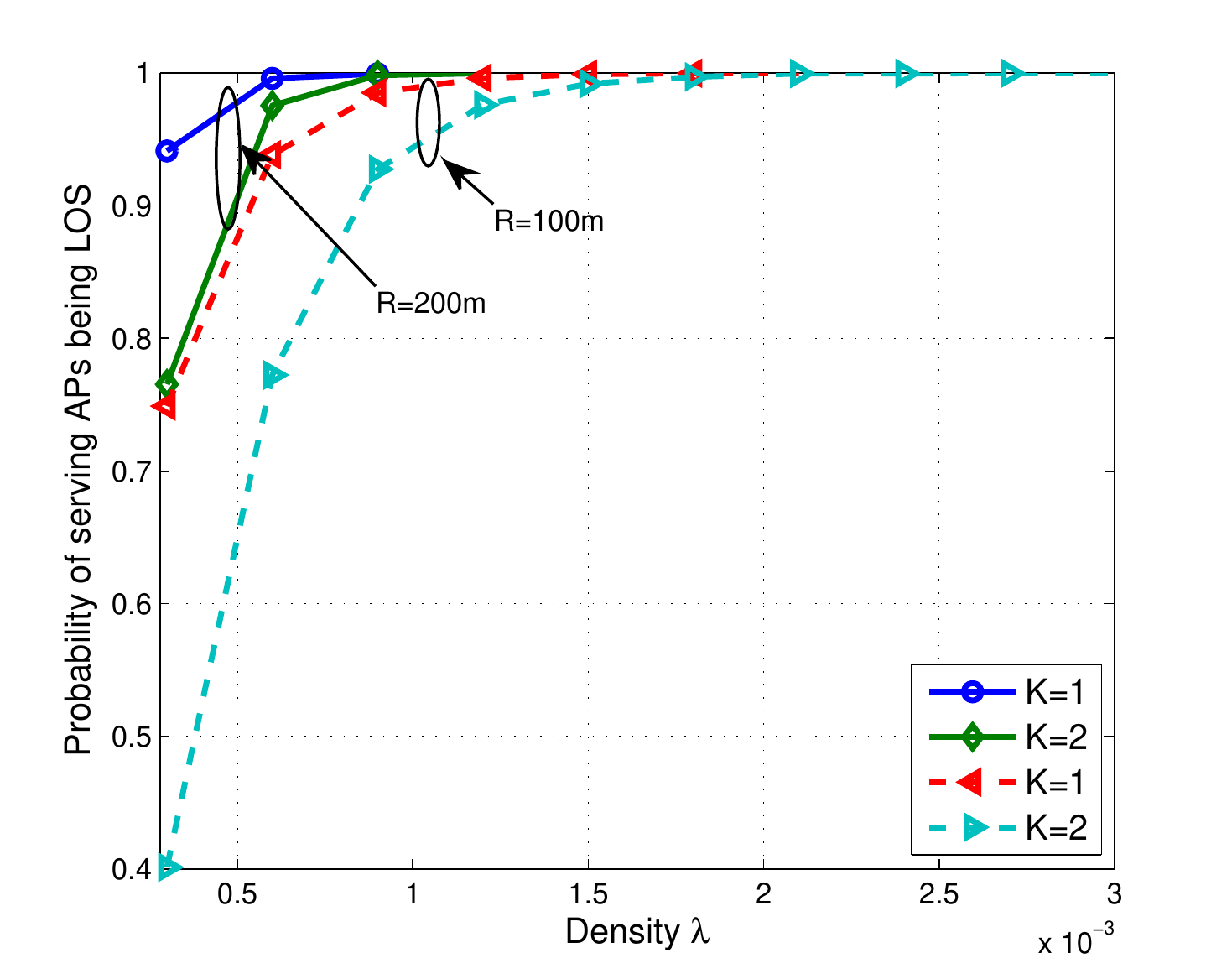}
	\caption{The probability of the serving APs being in LOS with different density $\lambda$ and zone radius $R$.}\label{fig_Los_pr_dsity}
\end{figure}

Recall that the serving APs for the UE are assumed to be in LOS with dense deployment. To justify this assumption, we present the probability of the serving APs being in LOS versus the AP density in Fig.~\ref{fig_Los_pr_dsity}.
From this figure, we can find that the probability increases with density and tends to be 1, which verifies our assumption.
	In addition, the probability increases with the increase of $R$ and decrease of $K$.
	The reason is that the average number of APs $N$ increases greatly with $R$, due to $N=\lambda\pi R^2$.
	Therefore, the distances between the serving APs to the UE are reduced, which increase the probability.

\begin{figure}
	\centering
	\includegraphics[width=3.2in]{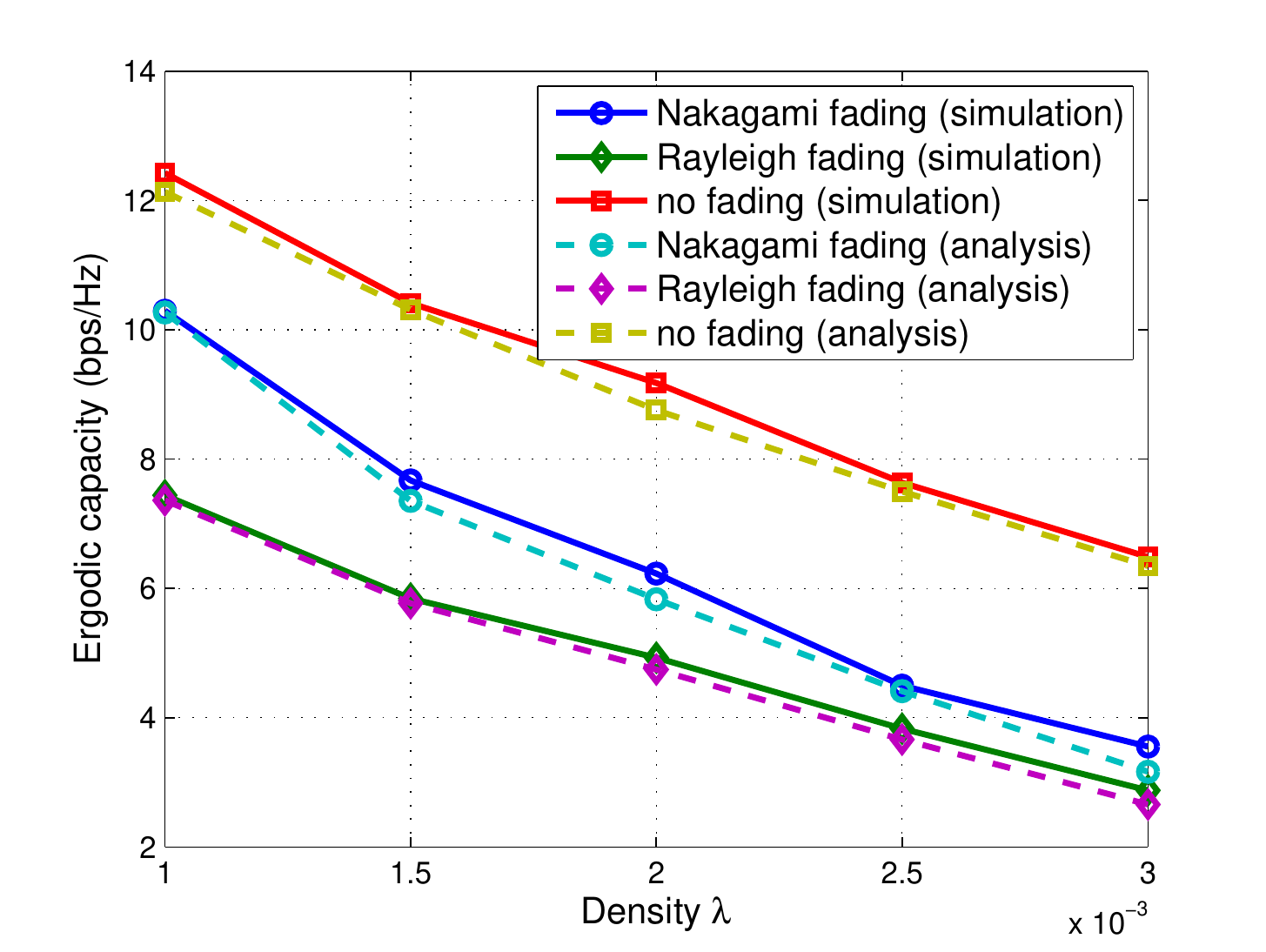}
	\caption{Ergodic capacity versus AP density under three small-scale fading distributions, where $K=2$.\label{fig_Cap_density}}
\end{figure}

\begin{figure}
	\centering
	\includegraphics[width=3.2in]{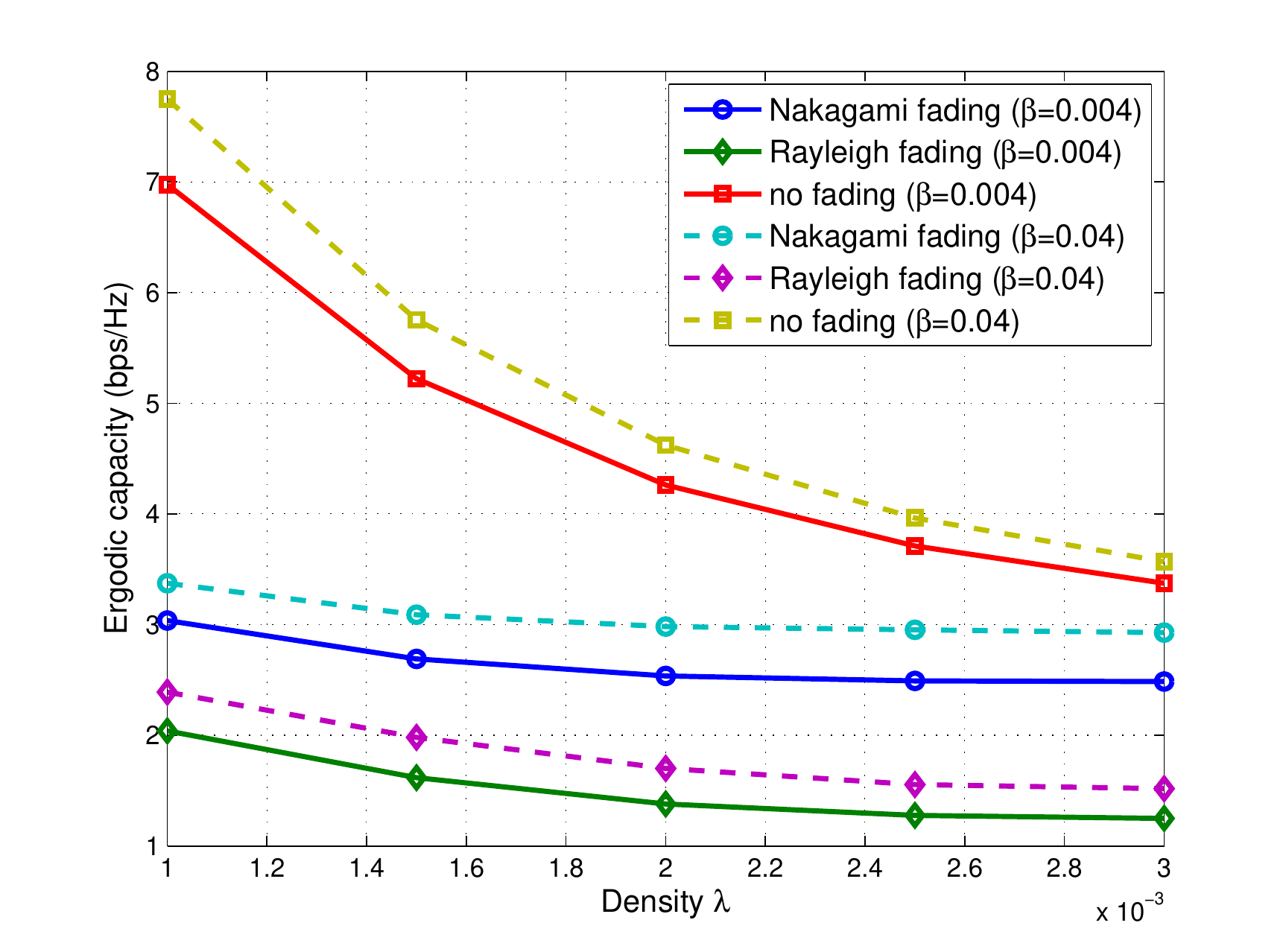}
	\caption{Ergodic capacity versus AP density under three small-scale fading distributions and two blockage parameters, where $K=1$.\label{fig_Cap_blockage}}
\end{figure}

Fig.~\ref{fig_Cap_density} shows the ergodic capacity with AP cooperation (i.e., $K=2$) under three small-scale fading distributions.
Firstly, the simulation results  match well with the analytical ergodic capacity expressions.
Secondly, as expected, the ergodic capacity is always the highest when ignoring the small-scale fading.
It is also interesting to find that the ergodic capacity under Nakagami fading is higher than that under Rayleigh fading.
Thirdly, the ergodic capacity decreases with the AP density $\lambda$, regardless of the fading model.
The reason is that, when the density increases, the number of interfering APs increases, while the number of serving APs remains unchanged.
Thus, the ergodic capacity is degraded.

Fig.~\ref{fig_Cap_blockage} shows the impact of the blockage parameter $\beta$ on the ergodic capacity.
It can be seen that the ergodic capacity increases with $\beta$ under all three fading distributions, which is consistent with \cite{Maamari2016coverage}.
This is because that the blockage probability of the LOS interfering links increases with $\beta$, and so does the number of NLOS APs.
As a result, the total interference power decreases and the ergodic capacity increases.

Fig.~\ref{fig_Cap_K} compares the ergodic capacity achieved under AP cooperation with that achieved without cooperation.
As expected, the scheme with AP cooperation achieves higher ergodic capacity than the one without cooperation.
In addition, AP cooperation provides higher capacity gain when the AP density is low.
This can be explained as: at low AP densities, the desired signal power, and hence the SINR, increase with the number of cooperative APs $K$, which improves the capacity performance.
When the AP density becomes large enough (e.g., $\lambda>0.0025$), the interference power dominates the desired power, offsetting the cooperative gain.

\begin{figure}
	\centering
	\includegraphics[width=3.2in]{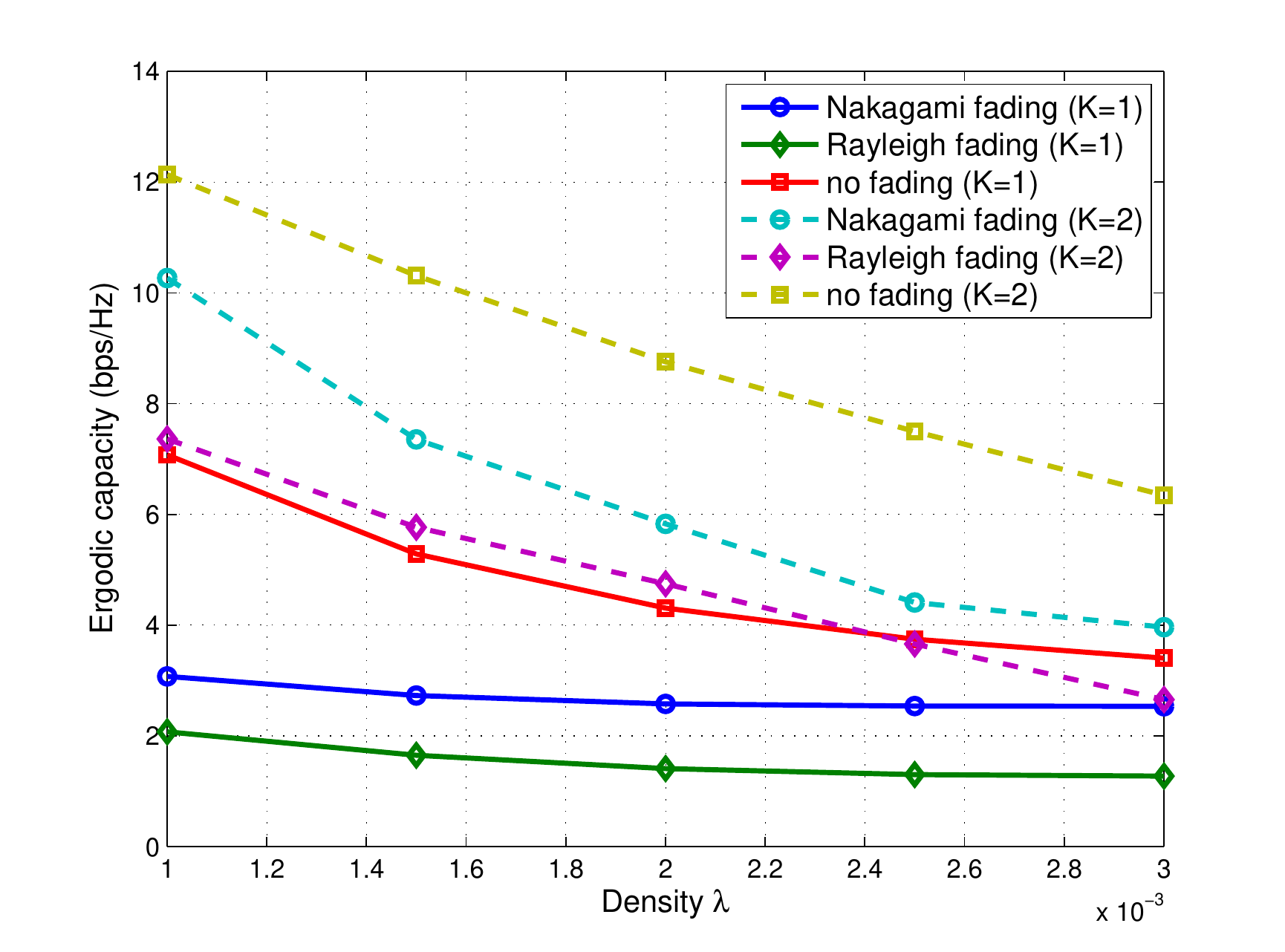}
	\caption{Ergodic capacity with two cooperation APs ($K=2$) and without AP cooperation ($K=1$).\label{fig_Cap_K}}
\end{figure}

\section{Conclusion}
\label{sec_clsn}

In this paper, we analyzed the ergodic capacity of the mmWave user-centric virtual cell dense networks where multiple APs and blockage structures are randomly distributed.
Taking into account the AP location randomness, directional beamforming, distance-dependent path loss and AP cooperation, we derived the analytical expressions for ergodic capacity under three small-scale fading distributions (i.e., Nakagami, Rayleigh and no fading).
Numerical results validated our analysis and showed that AP cooperation can provide distinct capacity gain, especially in a low-AP-density region.

\begin{appendices}
	\section{Proof of Theorem \ref{teom_C_cond_Nakagami}}\label{proof_th1}
	Conditioning on the serving APs being at distances $r_1\leqslant\cdots\leqslant r_K$ from the typical UE and the interfering APs being outside a circle of radius $r_K$, the conditional ergodic capacity is given by
	\begin{align}\label{poof_C_cond_Naka}
	C_{\text{cond}}(\bm{r})&=\mathbb{E}_{\mathbf{G},\bm \xi}\left[\log_2\left(1+\frac{\sum_{k\in\Phi\cap\mathcal{B}(r_K)}MC_{\rm L}r_k^{-\alpha_{\rm L}}|\xi_{k}|^2}{I_{\rm L}+I_{\rm N}+\sigma^2}\right)\right]\nonumber\\
	&\overset{(\rm a)}{=}\frac{1}{\ln2}\mathbb{E}\left[\int_0^\infty\frac{e^{-s\sigma^2}}{s}e^{-s(I_{\rm L}+I_{\rm N})}(1-e^{-sP})\right]\text{d} s\nonumber\\
	&\overset{(\rm b)}{=}\frac{1}{\ln2}\int_0^\infty\mathbb{E}\left[\frac{e^{-s\sigma^2}}{s}e^{-sI_{\rm L}}e^{-sI_{\rm N}}(1-e^{-sP})\right]\text{d} s\nonumber\\
	&\overset{(\rm c)}{=}\frac{1}{\ln2}\int_0^\infty\frac{e^{-s\sigma^2}}{s}\mathcal{L}_{I_{\rm L}}(s)\mathcal{L}_{I_{\rm N}}(s)(1-\mathcal{L}_{P}(s))\text{d} s,
	\end{align}
	where $P=\sum_{k\in\Phi\cap\mathcal{B}(r_K)}MC_{\rm L}r_k^{-\alpha_{\rm L}}|\xi_{k}|^2$, (a) follows from \cite[Lemma 1]{hamdi2008capacity}
	\begin{equation*}
	\ln(1+x)=\int_{0}^{\infty}\frac{1}{z}(1-e^{-xz})e^{-z}\text{d}z
	\end{equation*}
	with $z$ replaced by $s(I_{\rm L}+I_{\rm N}+\sigma^2)$.
	Step (b) is obtained by changing the order of integration and expectation and (c) follows from the fact that $I_{\rm L}$, $I_{\rm N}$ and $P$ are independent.
	$\mathcal{L}_{X}(s)$ is the Laplace Transformation (LT) of the PDF of r.v.$X$.
	
	Now we need to compute the LT of $P$, $I_{\rm L}$ and $I_{\rm N}$ in (\ref{poof_C_cond_Naka}).
	Firstly, let $p_k=MC_{\rm L}r_k^{-\alpha_{\rm L}}|\xi_{k}|^2, k=1,\cdots,K$.
	The LT of $p_k$ can be given by
	\begin{align*}
	\mathcal{L}_{p_k}(s)&=\mathbb{E}[e^{-sp_k}]=\mathbb{E}_{\xi_{k}}\left[e^{-sMC_{\rm L}r_k^{-\alpha_{\rm L}}|\xi_{k}|^2}\right]\\
	&\overset{(\rm a)}{=}\frac{1}{(1+sMC_{\rm L}r_k^{-\alpha_{\rm L}}/N_{\rm L})^{N_{\rm L}}},
	\end{align*}
	where $|\xi_{k}|^2$ is a normalized Gamma r.v. with parameter $N_{\rm L}$ and (a) is obtained by computing its LT.
	Because of the independence among $p_1,\cdots,p_K$, the PDF of $P=\sum_{k=1}^{K}p_k$ is given by
	\begin{equation*}
	f_{P}(z)=f_{p_1}(z)*\cdots*f_{p_K}(z),
	\end{equation*}
	where $*$ is convolution operation and $f_{p_k}(z)$ is the PDF of $p_k$.
	Thus, the LT of $P$ can be obtained as follows,
	\begin{equation}\label{LT_P_Naka}
	\mathcal{L}_P(s)=\prod_{k=1}^K\mathcal{L}_{p_k}(s)=\prod_{k=1}^K\frac{1}{(1+sMC_{\rm L}r_k^{-\alpha_{\rm L}}/N_{\rm L})^{N_{\rm L}}}.
	\end{equation}
	
	Secondly, the LT of $I_{\rm L}$ can be derived as follows,
	\begin{align}\label{LT_IL_Naka}
	\mathcal{L}_{I_{\rm L}}(s)
	&=\mathbb{E}_{\Phi_{\rm L},\mathbf{G},\bm{\xi}}\left[e^{-s\sum_{l\in \Phi_{\rm L}\cap\bar{\mathcal{B}}(0,r_K)}G_{l}C_{\rm L}r_l^{-\alpha_{\rm L}}|\xi_{l}|^2}\right]\nonumber\\
	&\overset{(\rm a)}{=}e^{-2\pi\lambda\sum_{n=1}^{2}b_n\int_{r_K}^{\infty}\left(1-\mathbb{E}_{\xi}\left[e^{-sa_nC_{\rm L}x^{-\alpha_{\rm L}}|\xi|^2}\right]\right)p(x)x\text{d}x}\nonumber\\
	&\overset{(\rm b)}{=}\prod_{n=1}^{2}e^{-2\pi\lambda b_n\int_{r_K}^{\infty}\left(1-1/(1+sa_nC_{\rm L}x^{-\alpha_{\rm L}}/N_{\rm L})^{N_{\rm L}}\right)p(x)x\text{d}x}\nonumber\\
	&=e^{-Q_{\rm L}(s)},
	\end{align}
	where $p(x)$ is the LOS probability function, $a_n$ and $b_n$ are defined in Subsection \ref{subsec_path_beam}; (a) follows from computing the Laplace function of the PPP $\Phi_{\rm L}$ \cite{baccelli2009stochastic}; (b) is obtained by computing the LT of $|\xi|^2$.
	
	In a similar way, for the NLOS interfering links, the small-scale fading term $|\xi_{l}|^2$ is a normalized Gamma r.v. with parameter $N_{\rm N}$.
	Thus, the LT of $I_{\rm N}$ is given by
	\begin{align}\label{LT_IN_Naka}
	&\mathcal{L}_{I_{\rm N}}(s)\nonumber\\
	&=\mathbb{E}_{\Phi_{\rm N},\mathbf{G},\bm{\xi}}\left[e^{-s\sum_{l\in \Phi_{\rm N}\cap\bar{\mathcal{B}}(0,r_K)}G_{l}C_{\rm N}r_l^{-\alpha_{\rm N}}|\xi_{l}|^2}\right]\nonumber\\
	&=e^{-2\pi\lambda\sum_{n=1}^{2}b_n\int_{r_K}^{\infty}\left(1-\mathbb{E}_{\xi}\left[e^{-sa_nC_{\rm N}x^{-\alpha_{\rm N}}|\xi|^2}\right]\right)(1-p(x))x\text{d}x}\nonumber\\
	&=\prod_{n=1}^{2}e^{-2\pi\lambda b_n\int_{r_K}^{\infty}\left(1-1/(1+sa_nC_{\rm N}x^{-\alpha_{\rm N}}/N_{\rm N})^{N_{\rm N}}\right)(1-p(x))x\text{d}x}\nonumber\\
	&=e^{-Q_{\rm N}(s)}.
	\end{align}
	Then, (\ref{fmula_C_cond_Naka}) is obtained by substituting (\ref{LT_P_Naka}) (\ref{LT_IL_Naka}) and (\ref{LT_IN_Naka}) into (\ref{poof_C_cond_Naka}).

	\section{Proof of Theorem \ref{teom_C_cond_Rayleigh}}\label{proof_th2}
	Based on the (\ref{poof_C_cond_Naka}) in Appendix \ref{proof_th1}, we only need to compute the LT of $I_{\rm L}$, $I_{\rm N}$ and $P$ with Rayleigh fading in the following.
	
	Let $g_k=|\xi_{k}|^2$ so that $p_k=MC_{\rm L}r_k^{-\alpha_{\rm L}}g_k, k=1,\cdots,K$.
	\begin{align*}
	\mathcal{L}_{p_k}(s)&=\mathbb{E}_{g_{k}}\left[e^{-sMC_{\rm L}r_k^{-\alpha_{\rm L}}g_k}\right]\\
	&\overset{(\rm a)}{=}\frac{1}{1+\mu sMC_{\rm L}r_k^{-\alpha_{\rm L}}},
	\end{align*}
	where (a) follows from the fact that $g_k\sim \exp(\mu)$ and computing its LT.
	Then, similar to (\ref{LT_P_Naka}), the LT of $P$ with Rayleigh fading is given by
	\begin{equation}\label{LT_P_Ray}
	\mathcal{L}_P(s)=\prod_{k=1}^K\mathcal{L}_{p_k}(s)=\prod_{k=1}^K\frac{1}{1+\mu sMC_{\rm L}r_k^{-\alpha_{\rm L}}}.
	\end{equation}
	
	The LT of $I_{\rm L}$ with Rayleigh fading can be derived as
	\begin{align}\label{LT_IL_Ray}
	\mathcal{L}_{I_{\rm L}}(s)
	&=\mathbb{E}_{\Phi_{\rm L},\mathbf{G},\bm{\xi}}\left[e^{-s\sum_{l\in \Phi_{\rm L}\cap\bar{\mathcal{B}}(0,r_K)}G_{l}C_{\rm L}r_l^{-\alpha_{\rm L}}|\xi_{l}|^2}\right]\nonumber\\
	&\overset{(\rm a)}{=}e^{-2\pi\lambda\sum_{n=1}^{2}b_n\int_{r_K}^{\infty}\left(1-\mathbb{E}_{g}\left[e^{-sa_nC_{\rm L}x^{-\alpha_{\rm L}}g}\right]\right)p(x)x\text{d}x}\nonumber\\
	&\overset{(\rm b)}{=}\prod_{n=1}^{2}e^{-2\pi\lambda b_n\int_{r_K}^{\infty}\left(1-1/(1+\mu sa_nC_{\rm L}x^{-\alpha_{\rm L}})\right)p(x)x\text{d}x}\nonumber\\
	&=e^{-V_{\rm L}(s)},
	\end{align}
	where $a_n$ and $b_n$ are defined in Subsection \ref{subsec_path_beam}; (a) is obtained by computing the Laplace function of the PPP $\Phi_{\rm L}$ \cite{baccelli2009stochastic}; (b) is obtained by computing the LT of $g$.
	
	Similarly, the LT of $I_{\rm N}$ with Rayleigh fading is given by
	\begin{align}\label{LT_IN_Ray}
	&\mathcal{L}_{I_{\rm N}}(s)\nonumber\\
	&=\prod_{n=1}^{2}e^{-2\pi\lambda b_n\int_{r_K}^{\infty}\left(1-1/(1+\mu sa_nC_{\rm N}x^{-\alpha_{\rm N}})\right)(1-p(x))x\text{d}x}\nonumber\\
	&=e^{-V_{\rm N}(s)}.
	\end{align}
	Then, (\ref{fmula_C_cond_Ray}) is obtained by substituting (\ref{LT_P_Ray}) (\ref{LT_IL_Ray}) and (\ref{LT_IN_Ray}) into (\ref{poof_C_cond_Naka}).
	
		\vspace{-1.0em}
	
	\section{Proof of Theorem \ref{teom_C_cond_nofading}}\label{proof_th3}
	Based on (\ref{SINR_nofading}), the conditional ergodic capacity without fading can be derived as
	\begin{align}\label{poof_C_cond_nofading}
	&C_{\text{cond}}(\bm{r})\nonumber\\
	&=\mathbb{E}_{\mathbf{G}}\left[\log_2\left(1+\frac{\sum_{k\in\Phi\cap\mathcal{B}(r_K)}MC_{\rm L}r_k^{-\alpha_{\rm L}}}{I'_{\rm L}+I'_{\rm N}+\sigma^2}\right)\right]\nonumber\\
	&=\frac{1}{\ln2}\int_0^\infty\frac{e^{-s\sigma^2}}{s}\mathcal{L}_{I'_{\rm L}}(s)\mathcal{L}_{I'_{\rm N}}(s)\left(1-\prod_{k=1}^{K}e^{-sMC_{\rm L}r_k^{-\alpha_{\rm L}}}\right)\text{d} s.
	\end{align}
	The detailed derivations of $\mathcal{L}_{I'_{\rm L}}(s)$ and $\mathcal{L}_{I'_{\rm N}}(s)$ are similar to (\ref{LT_IL_Naka}) and (\ref{LT_IN_Naka}), respectively.
	Thus, we only present the final results as follows,
	\begin{align}
	\mathcal{L}_{I'_{\rm L}}(s)
	&=e^{-2\pi\lambda\sum_{n=1}^{2} b_n\int_{r_K}^{\infty}\left(1-e^{-sa_nC_{\rm L}x^{-\alpha_{\rm L}}}\right)p(x)x\text{d}x}\nonumber\\
	&=e^{-W_{\rm L}(s)},\label{LT_IL_nofading}\\
	\mathcal{L}_{I'_{\rm N}}(s)
	&=e^{-2\pi\lambda\sum_{n=1}^{2} b_n\int_{r_K}^{\infty}\left(1-e^{-sa_nC_{\rm N}x^{-\alpha_{\rm N}}}\right)(1-p(x))x\text{d}x}\nonumber\\
	&=e^{-W_{\rm N}(s)}.\label{LT_IN_nofading}
	\end{align}
	Then, (\ref{fmula_C_cond_nofading}) is obtained by substituting (\ref{LT_IL_nofading}) and (\ref{LT_IN_nofading}) into (\ref{poof_C_cond_nofading}).
	
\end{appendices}


\section*{Acknowledgment}

This work was supported by the National Science and Technology
Major Project (2016ZX03001016-003) and the National Natural
Science Foundation of China under grant Nos. 61372106 and
61221002.
J. Shi acknowledges the support of the China Scholarship Council for a 1-year fellowship to McGill University.
B. Champagne acknowledges the financial support of NSERC of Canada.



%
\bibliographystyle{IEEEtran}
\bibliography{IEEEabrv,MMM_PAmmWave}

\end{document}